\documentclass[11pt]{article}

\usepackage{amsthm, amsmath, amssymb}
\usepackage{hyperref}
\usepackage{enumitem}
\usepackage{authblk}

\usepackage{graphicx}
\usepackage{color}

\usepackage{complexity}

\newtheorem{theorem}{Theorem}
\newtheorem{lemma}{Lemma}
\newtheorem{observation}[lemma]{Observation}
\newtheorem{definition}[lemma]{Definition}
\newtheorem{corollary}[lemma]{Corollary}
\newtheorem*{lem:circulationFaces}{Lemma~\ref*{lem:circulationFaces}}
\newtheorem*{lem:faceToEdge}{Lemma~\ref*{lem:faceToEdge}}

\newcommand{\abs}[1]{\lvert #1 \rvert}
\newcommand{\chev}[1]{\langle #1 \rangle} 
\renewcommand{\P}{\mathcal{P}}
\newcommand{\Z}{\mathbb Z}
\newcommand{\N}{\mathbb N}
\newcommand{\T}{\mathcal T}
\newcommand{\Gt}{\chev{\P_c}_3}

\newlang{\searchPM}{\textsc{Search-PM}}
\newlang{\decisionPM}{\textsc{Decision-PM}}
\newlang{\minwtPM}{\textsc{Min-Weight-PM}}
\newlang{\countPM}{\textsc{Count-PM}}
\newlang{\FL}{FL}

\renewcommand{\path}{\operatorname{path}}
\newcommand{\wt}{\mathsf{wt}}

\newcommand{\plainabstract}[1]{}

\allowdisplaybreaks

\begin{document}

\title{Derandomizing Isolation Lemma for $K_{3,3}$-free and $K_5$-free Bipartite Graphs}
\author{Rahul Arora}
\author{Ashu Gupta} 
\author{Rohit Gurjar}
\author{Raghunath Tewari}
\affil{{\footnotesize arorar@iitk.ac.in, ashug@iitk.ac.in, rgurjar@iitk.ac.in, rtewari@iitk.ac.in}\\
Indian Institute of Technology Kanpur}
\date{\vspace{-15mm}}
\maketitle
\begin{abstract}
The perfect matching problem has a randomized $\NC$ algorithm, using the celebrated
Isolation Lemma of Mulmuley, Vazirani and Vazirani.
The Isolation Lemma states that giving a random weight assignment to the edges of a graph,
ensures that it has a unique minimum weight perfect matching, with a good probability. 
We derandomize this lemma for $K_{3,3}$-free and $K_5$-free bipartite graphs, i.e.\
we give a deterministic log-space construction of such a weight assignment for these graphs. 
Such a construction was known previously for planar bipartite graphs. 
Our result implies that the perfect matching problem 
for $K_{3,3}$-free and $K_5$-free bipartite graphs
is in $\SPL$.
It also gives an alternate proof for an already known result -- reachability 
for $K_{3,3}$-free and $K_5$-free graphs is in $\UL$.
\end{abstract}

\plainabstract{
The perfect matching problem has a randomized NC algorithm, using the celebrated
Isolation Lemma of Mulmuley, Vazirani and Vazirani. The Isolation Lemma states that giving a random weight assignment to the edges of a graph, ensures that it has a unique minimum weight perfect matching, with a good probability. 
We derandomize this lemma for K_{3,3}-free and K_5-free bipartite graphs, i.e.
we give a deterministic log-space construction of such a weight assignment for these graphs. Such a construction was known previously for planar bipartite graphs. 
Our result implies that the perfect matching problem for K_{3,3}-free and K_5-free bipartite graphs is in SPL. It also gives an alternate proof for an already known result -- reachability for K_{3,3}-free and K_5-free graphs is in UL.
}

\section{Introduction}
The perfect matching problem is one of the most extensively studied problem
in combinatorics, algorithms and complexity.
In complexity theory, the problem plays a crucial role in the study of 
parallelization and derandomization. 
In a graph $G(V,E)$, a {\em matching} is a set of disjoint edges and 
a matching is called {\em perfect} if it covers all the vertices of the graph.
Edmonds \cite{Edm65} gave the first polynomial time algorithm for the matching problem. 
Since then, there have been improvements in its sequential complexity \cite{MV80},
but an $\NC$ (efficient parallel) algorithm for it is not known.
The perfect matching problem has various versions: 
\begin{itemize}[noitemsep]
\vspace{-5pt}
\item $\decisionPM$: Decide if there exists a perfect matching in the given graph. 
\item $\searchPM$: Construct a perfect matching in the given graph, if it exists.
\end{itemize}

A randomized $\NC$ ($\RNC$) algorithm for $\decisionPM$ was given by \cite{Lov79}. 
Subsequently, $\searchPM$ was also shown to be in $\RNC$ \cite{KUW86,MVV87}.
The solution of Mulmuley et al.\ \cite{MVV87} was based on the powerful idea of {\em Isolation Lemma}. 
They defined a notion of an isolating weight assignment on the edges of a graph. 
Given a weight assignment on the edges, weight of a matching
$M$ is defined to be the sum of the weights of all the edges in it.  

\begin{definition}[\cite{MVV87}]
For a graph $G(V,E)$, a weight assignment $w \colon E \to \N$ is isolating
if 
there exists a unique minimum weight perfect matching in $G$, according to $w$.
\end{definition}

The Isolation Lemma states that a random weight assignment (polynomially bounded) is isolating with a
good probability. Other parts of the algorithm in \cite{MVV87} are deterministic. 
They showed that if 
we are given an isolating weight assignment 
(with polynomially bounded weights) for a graph $G$ then
a perfect matching in $G$ can be constructed in $\NC^2$.
Later, Allender et al.\ \cite{ARZ99} showed that the $\decisionPM$ is in $\SPL$, 
if an isolating weight assignment can be constructed in $\L$ (see also \cite{DKR10}).
A language $L$ is in class $\SPL$ if its characteristic function $\chi_L \colon \Sigma^* \to \{0,1\}$ can 
be (log-space) reduced to computing determinant of an integer matrix. 

Derandomizing the Isolation Lemma remains a challenging open question. 
It has been derandomized for some special classes of graphs: planar bipartite graphs \cite{DKR10,TV12},
 constant genus bipartite graphs \cite{DKTV12}, 
graphs with small number of matchings \cite{GK87,AHT07} and 
 graphs with small number of nice cycles \cite{Hoa10}.
A graph $G$ is bipartite if its vertex set can be partitioned into two parts $V_1, V_2$
such that any edge is only between a vertex in $V_1$ and a vertex in $V_2$. 
A graph is planar if it can be drawn on a plane without any edge crossings. 

We make a further step towards the derandomization of Isolation Lemma. 
We derandomize it for $K_{3,3}$-free bipartite graphs and $K_{5}$-free bipartite graphs. 
These classes are generalizations of planar bipartite graphs. 
For a graph $H$, $G$ is an $H$-free graph if $H$ is not a minor of $G$.
$K_{3,3}$ is the complete bipartite graph with $(3,3)$ nodes and $K_5$ is the complete
graph with $5$ nodes.
A planar graph is simultaneously $K_{3,3}$-free and $K_5$-free.

\begin{theorem}
For a $K_{3,3}$-free or $K_5$-free bipartite graph, an isolating weight
assignment (polynomially bounded) can be constructed in log-space.
\end{theorem}

This theorem together with the results of Allender et al.\ \cite{ARZ99} and
Datta et al.\ \cite{DKR10} gives us the following
results about matching.

\begin{corollary}
For a $K_{3,3}$-free or $K_5$-free bipartite graph,
\begin{itemize}[noitemsep]
\vspace{-5pt}
\item $\decisionPM$ is in $\SPL$.
\item $\searchPM$ is in $\FL^{\SPL}$.
\item $\minwtPM$ is in $\FL^{\SPL}$. 
\end{itemize}
\end{corollary}
Here, $\FL^{\SPL}$ refers to a log-space transducer with access to an $\SPL$ oracle.
The problem $\minwtPM$ asks to construct the minimum weight perfect matching
in a given graph with polynomially bounded weights on its edges. 

For $K_{3,3}$-free bipartite and $K_5$-free bipartite graphs,
an $\NC$ algorithm for $\searchPM$ was known. This is implied by combining
two results: (i) $\countPM$ (counting the number of perfect matchings)
 is in $\NC$ for $K_{3,3}$-free graphs \cite{Vaz89} and $K_5$-free graphs \cite{STW14}
(ii) $\searchPM$ $\NC$-reduces to $\countPM$ for bipartite graphs \cite{KMV08}.
The limitation of this idea is that $\countPM$ is $\# \P$-hard for general bipartite graphs.
Thus, there is no hope of generalizing this approach to work for all graphs. 
While, our ideas can potentially lead to a solution for general/bipartite graphs.

After our work, small genus bipartite graphs is the only remaining class of 
bipartite graphs for which $\countPM$ is in $\NC$ \cite{GL99, KMV08},
but construction of an isolating weight assignment is not known.

\paragraph{Main Idea:} 
We start with the idea of Datta et al.\ \cite{DKR10} which showed that 
nonzero circulation (weight in a fixed orientation) 
for every nice cycle implies isolation of a perfect matching.
To achieve nonzero circulation in a $K_{3,3}$-free or $K_5$-free graph, 
we work with its $3$-connected or $4$-connected component decomposition given by \cite{Wag37, Asa85}
(can be constructed in log-space \cite{TW14,STW14}). 
The components are either planar or constant-sized.
These components form a tree structure, when components are viewed as a node and there is an
edge between two components if they share a separating pair/triplet. 
For any cycle $C$, we break it into its fragments contained within each of these components,
which we call {\em projections} of $C$. 
These projections themselves are cycles.

Circulation of any cycle can be seen as a sum of circulations of its projections.
The components, where a cycle has a non-empty projection, form a subtree of the component tree.
The idea is to assign weights such that there is `central' node in this subtree 
which gets a weight higher than the total weight coming from other nodes in the subtree.
Weights within a component are given by modifying the
 already known techniques for planar graphs \cite{DKR10, Kor09, TV12} and constant sized graphs.

This idea would work only if the component tree has a small depth, which might not be true in general.
Thus, we create an $O(\log n)$-depth working tree, which has the same nodes as the component tree
but the edge relations are different. The working tree `preserves' the subtree structure in some sense.
This working tree can be constructed using the standard recursive procedure for finding a set of centers.
But, a log-space implentation needed a non-trivial idea (Section~\ref{sec:workingTree}).

As there are $O(\log n)$ levels, we need to ensure that at every level the total weight
gets multiplied by only a constant. 
Thus, in a planar component, every edge cannot be assigned a weight on a higher scale. 
Instead, we choose only those edges which surround a separating pair/triplet, 
and scale their weight by the total weight coming from the subtree attached at that separating pair/triplet.

Achieving non-zero circulation in log-space also puts directed rechability in $\UL$ \cite{RA00,BTV09,TV12}.
Thus, we get an alternate proof for the result -- 
directed reachability for $K_{3,3}$-free and $K_5$-free graphs is in $\UL$ \cite{TW14}. 

In Section~\ref{sec:prelim}, we introduce the concepts of nonzero circulation, 
clique-sum, graph decomposition and the corresponding component tree. 
In Section~\ref{sec:nonzeroCirc}, we give a logspace constrcution of a weight assignment
with nonzero circulation for every cycle, for a class of graphs defined via clique-sum operations
on planar and constant-sized graphs. 
In Section~\ref{sec:reductions}, we argue that $K_{3,3}$-free and $K_5$-free graphs fall into this class.  

\section{Preliminaries}
\label{sec:prelim}

Let us first define a skew-symmetric weight function on the edges of a graph. 
For this, we consider the edges of the graph directed in both the directions.
We call this directed set of edges $\vec{E}$. 
A weight function $w \colon \vec{E} \to \Z$ is called skew-symmetric if for 
any edge $(u,v)$, $w(u,v) = -w(v,u)$. 
\begin{definition}
For a cycle $C$, whose edges are given by $\{(v_1, v_2), (v_2,v_3), \dots, (v_{k-1},v_k), (v_k,v_1)\}$,
 its circulation
is defined to be $w(v_1,v_2) + w(v_2,v_3) + \dotsm + w(v_k, v_1)$.
\end{definition}
Clearly, as our weight function is skew-symmetric, changing the orientation of the cycle,
only changes the sign of the circulation. 
The following lemma \cite[Theorem 6]{TV12} gives the connection between nonzero circulations and isolation 
of a matching. For a bipartite (undirected) graph $G(V_1,V_2,E)$, 
a skew-symmetric weight function $w \colon \vec{E} \to \Z$ on its edges,
has a natural interpretation on the undirected edges 
as $\mathrm{w} \colon E \to \Z$ such that $\mathrm{w}(u,v) = w(u,v)$, where $ u \in V_1$ and $v\in V_2$.

\begin{lemma}[\cite{TV12}]
Let $w \colon \vec{E} \to \Z$ is skew-symmetric weight function on the edges of a bipartite
graph $G$ such that every cycle has a non-zero circulation. 
Then, $\mathrm{w} \colon E \to \Z$ is an isolating weight assignment for $G$.
\end{lemma}

The bipartiteness assumption is needed only in the above lemma.
We will construct a skew-symmetric weight function that 
guarantees nonzero circulation for every cycle,
for $K_{3,3}$-free and $K_5$-free graphs, i.e.\ without assuming bipartiteness.


\subsection{Clique-Sum}
We will show construction of a nonzero circulation weight assignment for a special class of graphs,
defined via a graph operation called {\em clique-sum}. 

\begin{definition}[Clique-Sum]
Let $G_1$ and $G_2$ be two graphs each containing a clique (of same size).
A clique-sum of graphs $G_1$ and $G_2$ is obtained 
from their disjoint union by identifying pairs of vertices in these two cliques
to form a single shared clique,
and by possibly deleting some of the edges in the clique. 
It is called a $k$-clique-sum if the cliques involved have at most $k$ vertices.
\end{definition}

One can form clique-sums of more than two graphs by a repeated application of 
clique-sum operation on two graphs (see Figure~\ref{fig:cliqueSum} in Appendix~\ref{sec:appendix}).
Using this, we define a new class of graphs. 

\begin{figure}
\centering
  \begin{picture}(0,0)%
  \includegraphics{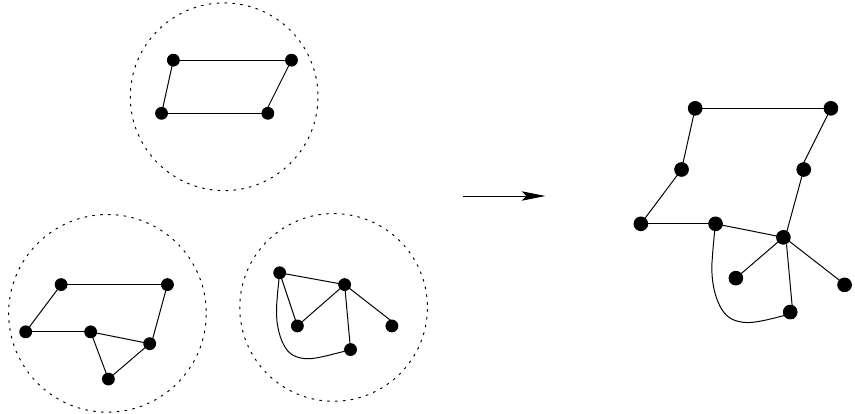}%
  \end{picture}%
  \input{cliqueSum.figtex}%

\caption{Graph $G$ obtained by taking (i) $2$-clique-sum of $G_1$ and $G_2$ by identifying
$\chev{u_1,v_1}$ with $\chev{u_2,v_2}$ and (ii)
$3$-clique-sum of the resulting graph with $G_3$ by identifying $\chev{a_2,b_2,c_2}$ with $\chev{a_3,b_3,c_3}$.
}
\label{fig:cliqueSum}
\end{figure}

Let $\P_c$ be the class of all planar graphs together with all graphs of size at most $c$, 
where $c$ is a constant. 
Define $\chev{\P_c}_k$ to be the class of graphs constructed by repeatedly taking 
$k$-clique-sums, starting from the graphs which belong to the class $\P_c$.
The starting graphs are called the component graphs.
We will construct a nonzero circulation weight assignment for the graphs which belong 
to the class $\Gt$.

Taking $1$-clique-sum of two graphs will result in a graph which is not biconnected.
As we are interested in perfect matchings, we only deal with biconnected graphs 
(see Section~\ref{sec:biconnected}).
Thus, we assume that every clique-sum operation
involves either $2$-cliques or $3$-cliques.
The $2$-cliques with respect to which we take cliques-sums are called separating pairs
and $3$-cliques are called separating triplets,
as their deletion will make the graph disconnected.
In general, they are called separating sets.
Usually, a separating pair/triplet means any pair/triplet of vertices,
whose deletion will make the graph disconnected. 
But, in this section, a separating pair/triplet will only mean those pairs/triplets 
which are used in a clique-sum operation. 

\subsection{Component Tree}
\label{sec:componentTree}
In general, clique-sum operation can be performed many times
using the same separating set. In other words, 
many components can share a separating set.
In Section~\ref{sec:reductions}, we show that any graph in $\Gt$ can be modified via some 
matching preserving operations such that on decomposition, 
any separating set is shared by only two components.
Henceforth, in this section we assume this property. 

Using this assumption, we can define a component graph for any graph $G \in \Gt$
as follows: each component is represented by a node and two such nodes are
connected by an edge if the corresponding components share a separating set. 
Observe that this component graph is actually a tree. This is because
when we take repeated clique-sums, a new component can be attached with only one of
the already existing components, as a clique will be contained within one component.
In literature \cite{HT73, TW14}, the component tree also contains a node for each separating set 
and it is connected by all the components which share this separating set. 
But, here we can ignore this node as we have only two sharers for each separating set. 

In the component tree, each component is shown with all the separating sets
it shares with other components. Thus, a copy of a separating set is present in
both its sharer components. Moreover, in each component, a separating set
is shown with a virtual clique, i.e.\ a virtual edge for a separating pair
and a virtual triangle for a separating triplet. 
These virtual cliques represent the paths between the nodes via other components
(see Figure~\ref{fig:componentTree}).
\begin{figure}
\centering
  \begin{picture}(0,0)%
  \includegraphics{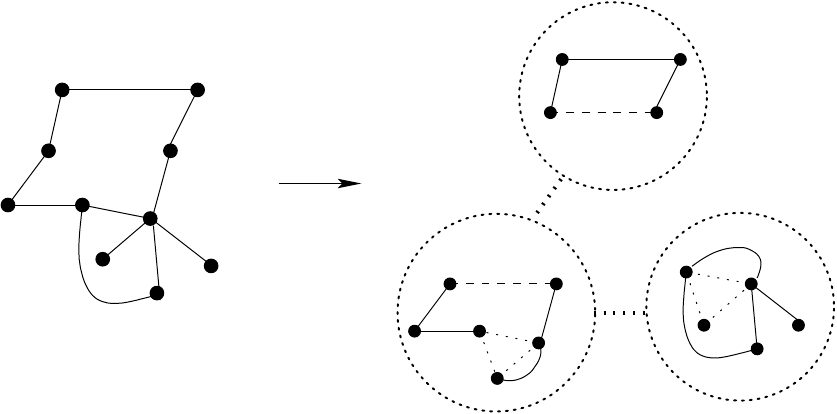}%
  \end{picture}%
  \input{componentTree.figtex}%

\caption{A graph $G \in \Gt$ is shown with its component tree. Dotted circles
show the nodes and dotted lines connecting them show the edges of the component tree. 
Dashed lines represent virtual edges and dotted triangles represent the virtual triangles, in the components. 
}
\label{fig:componentTree}
\end{figure}
If any two vertices in a separating set have a real edge in $G$,
then that real edge is drawn in one of the sharing components, parallel to the virtual edge. 
Note that while a vertex can have its copy in two components, 
any real edge is present in exactly one component. 

In literature \cite{HT73, TW14}, 
for any real edge in a separating set,
the component tree contains a new node called ``$3$-bond'' (a real edge with two parallel virtual edges).
But, here we do not have this node and represent the real
edge as mentioned above. 

\section{Nonzero Circulation}
\label{sec:nonzeroCirc}
In this section, we construct a nonzero circulation weight assignment for
a given graph in the class $\Gt$, provided that the component tree and 
the planar embeddings of the planar components are given. 
Moreover, to construct this weight assignment we will make some 
assumptions about the given graph and its component tree. 

\begin{enumerate}[noitemsep]
\vspace{-5pt}
\item In any component, a vertex is a part of at most one separating set.
\item Each separating set is shared by at most two components. 
\item Any virtual triangle in a planar component is always a face (in the given planar embedding).
\end{enumerate}

In Section~\ref{sec:reductions} we show how to construct a component tree for a
 given $K_{3,3}$-free or $K_5$-free graph and then to modify it
to have these properties.
The third property comes naturally, as the inside and outside parts of
 any virtual triangle can be considered as different components sharing this separating 
triplet. 
All these constructions are in log-space.  
Let, in any non-planar component, the number of real edges is bounded by $m$.
In Section~\ref{sec:reductions} we show that this bound is $60$, for a $K_{3,3}$-free or $K_5$-free graph.


\subsection{Components of a cycle}
We look at a cycle in the graph as {\em sum} of many cycles, one from each component 
the cycle passes through. Intuitively, the original cycle is {\em broken} at the separating set vertices which were part of the cycle, thereby generating fragments of the cycle in various nodes of the component tree. In all nodes containing these fragments, we include the virtual edges of the separating sets in question to complete the fragment into a cycle, thus resulting in component cycles in the nodes of the tree (see Figure~\ref{fig:cycleComponents}).

Consider a directed cycle $C = \{(v_0, v_1), (v_1, v_2), \dots, (v_{k-1}, v_0)\}$ in a graph $G = (V, E)$.
Without loss of generality, consider that $G$ is separated into two components $G_1$ and $G_2$ via a separating pair $(v_i, v_0)$ or a separating triplet $(v_i, v_0, u)$, where $1\le i<k$ and $u, v_0, \dots, v_k \in V$.
Then, one of the components, say $G_1$, will contain the vertices $v_i, v_{i+1 \bmod k}, \dots, v_{k-1}, v_0$, and the other ($G_2$) will contain the vertices $v_0, v_{1}, \dots, v_{i-1}, v_i$.
Then the cycles $C_1 = \{(v_i, v_{i+1 \bmod k}), \allowbreak\dots, \allowbreak(v_{k-1}, v_0), \allowbreak(v_0, v_i)\}$ and $C_2 = \{(v_0, v_{1 }), \allowbreak\dots, \allowbreak(v_{i-1}, v_i), \allowbreak(v_i, v_0)\}$ in $G_1$ and $G_2$ respectively are the component cycles of $C$, and we say that $C$ is the sum of $C_1$ and $C_2$.
Observe that the edges $(v_i, v_0)$ and $(v_0, v_i)$ are virtual.

Repeat the processes recursively for $C_1$ and $C_2$ until no separating set breaks a cycle component, and we get the component cycles of the cycle $C$. 
Note that any edge in $C$ is contained in one and only one of the component cycles, and for any component cycle, all its edges, other than the virtual edges, are contained in $C$.
\begin{figure}
\centering
  \begin{picture}(0,0)%
  \includegraphics{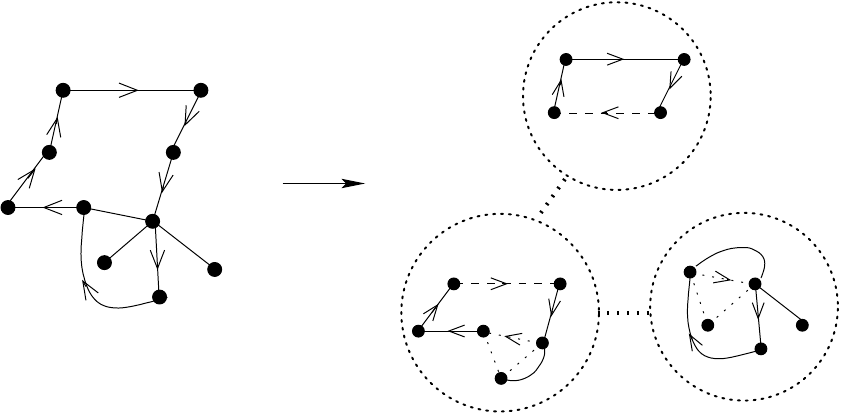}%
  \end{picture}%
  \input{cycleComponents.figtex}%

\caption{Breaking a cycle into its component cycles (projections) in the component tree. Notice that the original cycle and its components share the same set of {\em real} edges.
}
\label{fig:cycleComponents}
\end{figure}


Observe that for any separating set in a component, a cycle can use
one of its vertices to go out of the component and another vertex to come in
(this transition is represented by a virtual edge in the component). 
As any separting set has size at most $3$, a cycle can visit a node of the component tree only once. 
In other words, a cycle can have only one component cycle in any component tree node
(this would not be true if we had separating sets of size $4$). 
Also, a component cycle can take only one edge of any virtual triangle.

\begin{definition}[Projection of a cycle]
For a given component node $N$ in the component tree, the component cycle of a cycle $C$ in $N$ is called the projection of $C$ on $N$. If there is no component cycle of $C$ in $N$, then $C$ is said to have an empty projection on $N$.

\end{definition}
Within any component, weight of a virtual edge will always be set to zero. This is ensured by our weight function described in Section~\ref{sec:weightingScheme}.
Hence, the following lemma,
\begin{lemma}
The circulation of a cycle is the sum of the circulations of its component cycles.
\label{lem:breakCycle}
\end{lemma}
It is easy to see that for any cycle $C$, the components on which $C$ has a non-empty 
projection, form a subtree of the component tree.


\subsection{Weighting Scheme}
\label{sec:weightingScheme}
The actual weight function we employ is a combination of two weight functions $w_0$ and $w_1$.
They are combined with an appropriate scaling so that they do not interfere with each other. 
$w_1$ assures that all the cycles which are within one component have a non-zero circulation
and $w_0$ ensures that all the cycles which project on at least two components have a non-zero circulation.
We first describe the construction of $w_0$. 

\textbf{Working Tree:} The given component tree can have arbitrary depth,  
while our weight construction would need the tree-depth to be $O(\log n)$.
 Thus, we define a new {\em working tree}. 
It is a rooted tree, which
has the same nodes as the component tree, but the edge relations are different.
The working tree, in some sense, `preserves' the subtree structure of the original tree. 

For a tree $S$, its working tree $\wt(S)$ is constructed as follows: 
Find a `center' node $c(S)$ in the tree $S$ and 
mark it as the root of the working tree, $r(\wt(S))$.
Deleting the node $c(S)$ from the tree $S$, 
would give a set of disjoint trees, say $\{ S_1, S_2, \dots, S_k \}$.
Apply this procedure recursively on these trees to construct 
their working trees $\wt(S_1), \wt(S_2), \dots, \wt(S_k)$.
Connect each $\wt(S_i)$ to the root $r(\wt(S))$, as a subtree.
This completes the construction. 

The `center' nodes are chosen in a way so that the working tree depth is $O(\log n)$.
Section~\ref{sec:workingTree} gives the exact log-space construction of the working tree. 

Note that for any two nodes $v_1 \in S_i$ and $v_2 \in S_j$ such that $i\ne j$, 
$\path(v_1,v_2)$ in $S$ passes through the node $c(S) = r(\wt(S))$.
Thus, we get the following property for the working tree. 
\begin{observation}
\label{obs:path}
For any two nodes $u,v \in S$, let their least common ancestor in 
the working tree $\wt(S)$ be the node $a$. Then $\path(u,v)$ in the tree $S$
passes through $a$.
\end{observation}
The root $r(\wt(S))$ of the working tree $\wt(S)$ is said to be at level $1$. 
For any other node in $\wt(S)$, its level is defined to be one more than the level of its parent. 
Henceforth, level of a node will always mean its level in the working tree.
From Observation~\ref{obs:path}, we can easily conclude the following. 
\begin{observation}
\label{obv:subtree}
Let $S'$ be an arbitrary subtree of $S$, with its set of nodes being $\{v_1, v_2, \dots, v_k\}$.
There exists $i^* \in \{1,2, \dots, k \}$ such that for any $j \neq i^*$,
$v_{j}$ is a descendant of $v_{i^*}$ in the working tree $\wt(S)$.
\end{observation}
\begin{proof}
Let $l^*$ be the minimum level of any node in $S'$,
and let $v_{i^*}$ be a node in $S'$ with level $l^*$.
We claim that every other node in $S'$ is a descendant of $v_{i^*}$, in the working tree $\wt(S)$. 
For the sake of contradiction, let there be a node $v_j \in S'$, 
which is not a descendant of $v_{i^*}$.
Then, the least common ancestor of $v_j$ and $v_{i^*}$ in $\wt(S)$, must have a level,
strictly smaller than $l^*$. By observation~\ref{obs:path}, this least common
ancestor must be a present in the tree $S'$. But, we assumed $l^*$ is the minimum
level in $S'$. Thus, we get a contradiction. 
\end{proof}

This observation plays a crucial role in our weight assignment construction,
as for any cycle $C$ in the graph $G$, the nodes in the component tree,
where $C$ has a non-empty projection, form a subtree of the component tree. 


Complementary to the level, we also define {\em height} of every node in the working tree. 
Let the maximum level of any node in the working tree be $L$. 
Then, the height of a node is defined to be the difference between its level and $L+1$.

To assign weights in the graph $G$, we work with the working tree of its component
tree. Let the working tree be $\T$.
We start by assigning weight to the nodes having the largest level, 
and move up till we reach level $1$, that is, the root node $r(\T)$.



\paragraph{Circulation of cycles spanning multiple components:}
For any subtree $T$ of the working tree $\T$, the weights to the edges inside the component $r(T)$ will be given by two different schemes depending on whether the corresponding graph is planar or constant sized.

Let the maximum possible number of edges in a constant sized component be $m$. Then, let $K$ be a constant such that $K>\max{(2^{m+2}, 7)}$.
Also, suppose that the height of a node $N$ is given by the function $h(N)$, and the number of leaves in subtree $T$ is given by $l(T)$.
Lastly, suppose the set of subtrees attached at $r(T)$ is $\{ T_1,T_2, \dots, T_k \}$.

{\bf Constant sized graph:} 
Let the set of (real) edges of the graph is $\{ e_1, e_2,\allowbreak \dots, e_m \}$. 
The edge $e_j$ will be given weight $2^j\times K^{h(r(T))-1} \times l(T)$ for an arbitrarily fixed direction.
The intuition behind this scheme is that powers of $2$ ensure that sum of weights for any subset of edges remain nonzero even when they contribute with different signs. 
Later, we prove that for a cycle $C$ fully contained within a subtree $T$ of the working tree,
the upper bound on its circulation is $K^{h(r(T))}\times l(T)$.
 
{\bf Planar graph:} Let us fix a planar embedding of the graph. 
For a given weight assignment $w : \vec{E} \to \Z$ on the edges of the graph, 
we define the {\em circulation of a face} as the circulation of the corresponding
cycle in the clockwise direction i.e.\
traverse the boundary edges of the face in the clockwise direction 
and take the sum of their weights.
Here our weighting scheme will fix circulations for the inner faces of the graph.
Lemma~\ref{lem:faceToEdge} describes how to assign weights to the edges of a planar graph
to get the desired circulation for each of the inner faces.  

\paragraph{Assigning circulations to the faces:}
If $T$ is a singleton, and thus there are no subtrees attached at $T$, we give a zero circulation to all the faces (and thus to all the edges) of $r(T)$.

Otherwise, consider a separating pair $\{a, b\}$ where a subtree $T_i$ is attached to $r(T)$.
The two faces adjacent to the virtual edge $(a,b)$ will be assigned circulation $2\times K^{h(r(T_i))}\times l(T_i)$.
Similarly, consider a triplet $\{a, b, c\}$ where a subtree $T_j$ is attached.
Then all the faces (at most $3$) adjacent to the virtual triangle $\{a,b,c\}$ get circulation $2\times K^{h(r(T_j))}\times l(T_j)$.
Repeat this procedure for faces adjacent to all the pairs and/or triplets where subtrees are attached.
If a face is adjacent to more than one virtual edge/triangle, then
we just take the sum of different circulations due to each virtual edge/triangle.

Here, we mean that each face has a positive circulation in the clockwise direction.
The intuition behind this scheme is the following: circulation of any cycle in the planar component is just the sum of  circulations of the faces inside it.
As, all of them have same sign, they cannot cancel each other.
Moreover, contribution to the circulation from this planar component cannot be canceled by the contribution from any of its subtrees.

Now, we formally show that this weighting scheme ensures that all the cycles spanning multiple components in the tree get non-zero circulation.

\paragraph{Nonzero Circulation of a cycle:}
Firstly, we derive the upper bound $U_T$ on the circulation of any cycle completely contained in a subtree $T$ of the working tree.
\begin{lemma}
\label{lem:upperBound}
The upper bound on the circulation of any cycle contained in a subtree $T$ of the working tree $\T$ is $U_T = K^{h(r(T))}\times l(T)$.
\end{lemma}
\begin{proof}
We prove this using induction on the height of $r(T)$.

\textbf{Base case:}
The base case is when the height of $r(T)$ is $1$. Notice that this means that $r(T)$ has the maximum level amongst all the nodes in $\T$, and therefore, $r(T)$ is a leaf node, and $T$ is a singleton.
Consider the two cases: i)when $r(T)$ is a planar node, and ii)when it is a constant sized node.

By our weight assignment, if $r(T)$ is planar, the total weight of all the edges is zero.
On the other hand, if $r(T)$ is a constant sized graph, the maximum circulation of a cycle is the sum of weights of its edges, that is, $\sum_{i=1}^m (K^0\times 1\times 2^i) < 2^{m+1} \le K$.
Thus, the circulation is upper bounded by $K^{h(r(T))}\times l(T)$.

\textbf{Induction hypothesis:}
The upper bound for any tree $T'$ with $h(r(T'))\le j-1$ is $U_{T'} = K^{h(r(T'))}\times l(T')$. 

\textbf{Induction step:}
We will prove that the upper bound for any tree $T$, with $h(r(T))=j$, is $U_{T} = K^{h(r(T))}\times l(T)$.

Let the subtrees attached at $r(T)$ be $\{T_1, T_2, \dots, T_k\}$. For any cycle in $T$, sum of the circulations of its projections on the subtrees $T_1, T_2, \dots, T_k$ can be at most $\sum_{i=1}^k U_{T_i}$. 

First, we handle the case when $r(T)$ is planar. For any subtree $T_i$, the total circulation of faces in $r(T)$ due to connection to $T_i$ can be $6\times K^{h(r(T_i))}\times l(T_i)$. This is because the circulation of each face adjacent to the separating set connecting with $T_i$ is $2\times K^{h(r(T_i))}\times l(T_i)$, and there can be at most $3$ such faces.
Here, note that for all $i$, level of $r(T_i)$ is one more than level of $r(T)$, and thus height of $r(T_i)$ is one less than height of $r(T)$.
Thus,
\begin{align*}
U_T &= \sum_{i=1}^k U_{T_i} + \sum_{i=1}^k \left(6\times K^{h(r(T_i))}\times l(T_i)\right)\\
&= \sum_{i=1}^k \left(K^{h(r(T_i))}\times l(T_i)\right) + \sum_{i=1}^k \left(6\times K^{h(r(T_i))}\times l(T_i)\right)\\
&= \sum_{i=1}^k \left(7\times K^{h(r(T_i))}\times l(T_i)\right)\\
&= 7\times K^{h(r(T))-1}\times \sum_{i=1}^k l(T_i) && (\forall i,\ h(r(T_i)) = h(r(T))-1)\\
&< K^{h(r(T))} \times \sum_{i=1}^k l(T_i) && (K > 7)\\
&= K^{h(r(T))}\times l(T)
\end{align*}

Now, consider the case when $r(T)$ is a small non-planar graph. The maximum possible contribution from edges of $r(T)$ to the     circulation of a cycle in $T$ is less than $2^{m+1}\times K^{h(r(T))-1}\times l(T)$. Similar to the case when $r(T)$ is planar, contribution from all subtrees is at most $K^{h(r(T))-1}\times l(T)$. The total circulation of a cycle in $T$ can be at most the sum of these two bounds, and is thus bounded above by $(2^{m+1}+1)\times K^{h(r(T))-1}\times l(T)$. Since, $K>2^{m+2}$, the total possible circulation is less than $K^{h(r(T))}\times l(T)$.

Therefore, the upper bound $U_T = K^{h(r(T))}\times l(T)$.
\end{proof}

To see that each cycle gets a nonzero circulation, recall Lemma~\ref{lem:breakCycle}, which says that the circulation of the cycle is the sum of circulations of its projections on different components. 
Consider a cycle $C$.
We look at the minimum `level' component on which $C$ has a non-empty projection.
We show two things: (i) the contribution to the circulation from this component is nonzero, and
(ii) the contribution to the circulation from this component is larger than sum of all the circulation contributions from its higher level descendants in the working tree.

Observe that proving the above two will automatically prove that any cycle $C$ projecting on multiple component nodes has a non-zero circulation. This is because the nodes having non-empty projection from cycle $C$ form a subtree $S_C$ in the component tree; and when looking at the nodes of $S_C$ in the working tree $\mathcal{T}$, we can always find a node $v^*\in S_C$ such that all other nodes in $S_C$ are its descendants (see Observation \ref{obv:subtree}).
Let $v^*$ be the root of a subtree $T$ in the working tree. If the contribution from $v^*$~(or equivalently $r\left(T\right)$) to the cycle circulation is non-zero and exceeds the contribution from all its descendants, circulation of the cycle $C$ is certainly non-zero.

Again, let the subtrees attached at $r(T)$ be $\{T_1, T_2, \dots, T_k\}$.

Case 1: When the component is constant-sized.
It is easy to see that the circulation of any cycle in this component 
will be nonzero as long as it takes a real edge, because the weights given are powers of $2$. Also, the minimum weight of any edge in $r(T)$ is $2\times \sum_{i=1}^k U_{T_i}$. Thus, when a cycle takes a real edge, contribution to its circulation from $r(T)$ is larger than contribution from higher level components (components in the subtrees attached at $r(T)$).
And any cycle has to take a real edge, as the virtual edges and triangles 
all have disjoint set of vertices. (Here, the virtual triangle does not count as a cycle).


Case 2: When the component is planar. The crucial observation here is that all the faces inside a cycle contribute to its circulation in the same orientation.

\begin{lemma}
\label{lem:circulationFaces}
In a planar graph with a given planar embedding, circulation of a cycle in clockwise orientation is the 
sum of circulations of the faces inside it (Proof given in Appendix~\ref{sec:appendix}).
\end{lemma}


As all faces have positive circulation in clockwise direction, the total sum 
remains nonzero.
Now, observe that if the cycle $C$ goes through the subtree $T_i$, then its projection in $r(T)$, say $C_i$, must contain at least one of the faces adjacent to the pair/triplet in $r(T)$, at which $T_i$ is connected. Since, circulation of this face is $2U_{T_i}$, contribution from this component will surpass the total sum of all the subtrees where $C$ passes through.

Thus, we can conclude the following.
\begin{lemma}
\label{lem:multiCycleNonZero}
Circulation of any cycle which passes through at least two components is nonzero.
\end{lemma}

\textbf{Weights from faces to edges:}
Now, we come back to the question of assigning weights to the edges in a planar component such that
the faces get the desired circulations. 
Lemma~\ref{lem:faceToEdge} describes this procedure for any planar graph.
But, the scheme will assign weights to all the edges, 
while we are not allowed to give weights to virtual edges/triangles.
So, first we collapse all the virtual triangles to one node and all the
virtual edges to one node.
As no two virtual triangles/edges are adjacent, after this operation,
every face remains a non-trivial face (except the virtual triangle face).
Now, we apply the procedure from Lemma~\ref{lem:faceToEdge}. 
After undoing the collapse, the circulations of the faces will not change
and we will have the desired circulations. 

\begin{lemma}{\cite{Kor09}}
\label{lem:faceToEdge}
Let $G(V,E)$ be a planar graph with $F$ being its set of inner faces in some planar embedding.
For any given function on the inner faces $w' : F \to \Z$, a skew symmetric 
weight function $w \colon \vec{E} \to \Z$ can be constructed in log-space such that every face $f \in F$
has a circulation $w'(f)$ (Proof is described in Appendix~\ref{sec:appendix}).
\end{lemma}

\paragraph{Circulation of cycles contained within a single component:}
%
%
For planar components, to construct $w_1$, we assign $+1$ circulation to every face using Lemma~\ref{lem:faceToEdge}
(similar to the case of multiple components). This would ensure nonzero circulation for every cycle 
within the planar component. 
This construction has been used in \cite{Kor09} for bipartite planar graphs. 
\cite{TV12} also gives a log-space construction which ensures nonzero circulation for all cycles in a planar graph,
using Green's theorem.

For the non-planar components, $w_0$ already ensures that each cycle has non-zero circulation. Therefore, we set $w_1=0$.
Use a linear combination of $w_0$ and $w_1$ such that they do not interfere with each other. 
This together with Lemma~\ref{lem:multiCycleNonZero} gives us the following.
\begin{lemma}
Circulation of any cycle is non-zero.
\end{lemma}
\paragraph{Polynomially bounded weights:}
Now, we show that the weight given by this scheme is polynomially bounded.
\begin{lemma}
The total weight given by the weighting scheme is polynomially bounded. 
\end{lemma}
\begin{proof}
The weight $w_1$ is polynomially bounded according to the procedure in Lemma \ref{lem:faceToEdge}.

Consider $w_0$. Observe that the upper bound $U_\T$ for the circulation of a cycle in $\T$ is actually just the sum of weights of all the edges in constant sized components, and of all the faces in planar components.
Also, sum of the circulations of faces in a planar graph equals the sum of weight given to edges, by the construction given in the proof of Lemma \ref{lem:faceToEdge}.
Therefore, $U_\T$ gives the bound on the weight function $w_0$. Since the maximum level of any node in $\T$ can be at most $O(\log \abs \T)$, the height of $r(T)$, that is $h(r(T)) = O(\log\abs\T)$. Also, the total number of leaves in $\T$ is at most $\abs\T$.
\[U_\T = K^{h(r(\T))} \times l(\T) \le K^{O(\log\abs\T)}\times \abs\T = \abs\T^{O(\log K)}\abs\T = \abs\T^{O(\log K)}\]
If $n$ is the size of the original graph $G$, then clearly $\abs \T \le n$. Therefore, $U_\T~=~O(n^{O(\log K)})$. Recall that $K$ is a constant, and thus, $w_0$ is also polynomially bounded.

Since we use a linear combination of $w_0$ and $w_1$, the total weight function is polynomially bounded.
\end{proof}

\subsection{Construction of the Working Tree}
\label{sec:workingTree}
Now, we describe the log-space construction of the working tree. 
The idea is inspired from the construction of \cite[Lemma 6]{LMR07}, where
they create a $O(\log n)$-depth tree of well-matched substrings of a given well-matched string.
Recall that for a tree $S$, the working tree $\wt(S)$ is constructed by
first choosing a center node $c(S)$ of $S$
and marking it as the root of $\wt(S)$,
and then recursively finding the working trees for each component 
obtained by removing the node $c(S)$ from $S$ and connecting them to 
root of $\wt(S)$, as subtrees.

First consider the following possible definition of the center:
for any tree $S$ with $n$ nodes, one can define its 
center to be a node whose removal would give 
disjoint components of size $\leq 1/2 \abs{S}$. 
Finding such a center is an easy task and can be done in log-space.
Clearly, the depth of the working tree would be $O(\log n)$.
It is not clear if the recursive procedure of finding centers for
each resulting component can be done in log-space. 
Therefore, we give a more complicated way of defining the centers,
so that the whole recursive procedure can be done in log-space.

First, we make the tree $S$ rooted at an arbitrary node $r$.
To find the child-parent relations of the rooted tree, one can do 
the standard log-space traversal of a tree: for every node, give its edges
an arbitrary cyclic ordering. Start traversing from the root $r$ by taking an arbitrary edge.
If you arrive at a node $u$ using its edge $e$ then leave node $u$ using the right neighbor of $e$.
This traversal ends at $r$ with every edge being traversed exactly twice. 

For any node $v$, let $S_v$ denote the subtree of $S$, rooted at $v$. 
For any node $v$ and one of its descendant nodes $v'$ in $S$, 
let $S_{v,v'}$ denote the tree $S_v \setminus S_{v'}$.
Moreover $S_{v,\epsilon}$ would just mean $S_v$, for any $v$.
With our new definition of the center, 
at any stage of the recursive procedure, the component under consideration will 
always be of the form $S_{v,v'}$, for some nodes $v, v' \in S$.
Now, we give a definition of the center for a rooted tree of the form
$S_{v,v'}$.

\textbf{Center $c(S_{v,v'})$:}
case (i) When $v' = \epsilon$, i.e.\ the given tree is $S_v$.
Let $c$ be a node in $S_v$, such that its removal gives
components of size $\leq 1/2 \abs{S_v}$.
If there are more than one such nodes then choose the lexicographically smaller one
(there is at least one such center \cite{Jor69}).
Define $c$ as the center of $S_{v,v'}$.

Let the children of $c$ in $S_v$ be $\{c_1, c_2, \dots, c_k\}$.
Clearly, after removing $c$ from $S_v$, the components we get are
$S_{c_1}, S_{c_2}, \dots, S_{c_k}$ and $S_{v,c}$.
Thus, they are all of the desired form and have size $\leq 1/2 \abs{S_v}$.

case (ii) When $v'$ is an actual node in $S_{v}$. 
Let the node sequence on the path connecting $v$ and $v'$ be
 $(u_0, u_1, \dots, u_p)$,
with $u_0 = v$ and $u_p = v'$.
Let $0 \leq i \leq p$ be the least index such that
$\abs{S_{u_{i+1},v'}} \leq 1/2 \abs{S_{v,v'}}$.
This index exists because $\abs{S_{u_p,v'}}=0$.
Define $u_i$ as the center of $S_{v,v'}$.

Let the children of $u_i$, apart from $u_{i+1}$, be $\{c_1, c_2, \dots, c_k\}$.
After removal of $u_i$ from $S_{v,v'}$, the components we get are
$S_{c_1}, S_{c_2}, \dots, S_{c_k}$, $S_{u_{i+1},v'}$ and $S_{v,u_i}$.
By the choice of $i$, $\abs{S_{u_i, v'}} > 1/2 \abs{S_{v,v'}}$.
Thus, $\abs{S_{v,u_i}} \leq 1/2 \abs{S_{v,v'}}$.
So, the only components for which we do not have a guarantee on their sizes,
are $S_{c_1}, S_{c_2}, \dots, S_{c_k}$.
Observe that when we find a center for the tree $S_{c_j,\epsilon}$ in the next recursive call, 
it will fall into case (i) and
the components we get will have their sizes reduced by a factor of $1/2$.

Thus, we can conclude that in the recursive procedure for constructing
the working tree, we reduce the size of the component by half in at most two recursive calls.
Hence, the depth of working tree is $O(\log n)$.
Now, we describe a log-space procedure to construct the working tree. 

\begin{lemma}
\label{lem:workingTree}
For any tree $S$, its working tree $\wt(S)$ can be constructed 
in log-space. 
\end{lemma}
\begin{proof}
We just describe a log-space procedure for finding the parent of a given node $x$ 
in the working
tree. Running this procedure for every node will give us the working tree.

Find the center of the tree $S$. Removing the center would give many components. 
Find the component $S_1$, to which the node $x$ belongs.
Apply the same procedure recursively on $S_1$. Keep going to smaller components which contain $x$,
till $x$ becomes the center of some component. The center of the previous component
in the recursion will be the parent of $x$ in the working tree. 

In this recursive procedure, to store the current component $S_{v,v'}$, we just need to store
two nodes $v$ and $v'$. Apart from these, we need to store center of the previous component 
and size of the current component. 

To find the center of a given component $S_{v,v'}$, 
go over all possibilities of the center, depending on whether $v'$ is $\epsilon$ or a node.
For any candidate center $c$, find the sizes of the components generated if $c$ is removed.
Check if the sizes satisfy the specified requirements. 
Any of these components is also of the form $S_{u,u'}$ and thus can be stored with two nodes. 

By the standard log-space traversal of a tree (see, for example \cite{Lin92}), 
for any given tree $S_{v,v'}$, one can count the number of nodes in it and test membership
of a given node. 
Thus, the whole procedure works in log-space.
\end{proof}


\subsection{Complexity of the weight assignment}
\label{sec:timeComplexity}
We use simple log-space procedures in sequence to assign the weights in the working tree.
After construction of the working tree, we use iterative log-space procedures to store the following for each node: i) the level of the node, and ii) the number of leaves in the subtree rooted at it. Both just require tree traversal while keeping a counter, and can clearly be done in log-space.
Also, since we have the maximum level amongst all the nodes, we can use it in another straightforward log-space function to compute the height of every node.
We store one more piece of information. Let the subtrees of the component tree $S$ attached at a node $N$ be $S_1, S_2, \dots, S_k$.
Correspondingly, in the working tree, the children of $N$ will be $r(\wt(S_1)), r(\wt(S_2)), \dots, r(\wt(S_k))$.
For all $i\ (1\le i\le k)$, we remember which virtual edge/triangle of $N$ is shared with the subtree $S_i$ by storing a pointer to the node $r(\wt(S_i))$.

Next, we iterate on the nodes of the working tree to assign the weights.
For every non-planar component, we iterate on edges inside it in an arbitrary (deterministic) fashion, and assign a weight of $2^i\times K^{(h(N)-1)}\times l(T(N))$, where $i$ is the iteration count, $N$ is the node, and $T(N)$ is the subtree rooted at $N$.

In the next step, we again iterate on the nodes, and for every node $N$, we visit all its virtual edges/triangles. For a given virtual edge/triangle $\tau_i$, let the child of $N$ in the working tree attached at $\tau_i$ be $N_{i}$. We add a circulation of $2\times K^{h(N_i)}\times l(T(N_i))$ to all the faces adjacent to $\tau_i$.
As the last step, we find the weights for the edges which would give the desired circulations of the faces.
Lemma~\ref{lem:faceToEdge} shows that it can be done in log-space.


\section{$K_{3,3}$-free and $K_5$-free graphs}
\label{sec:reductions}
In this section, 
we show how to construct the desired component tree 
for any given $K_{3,3}$-free or $K_5$-free graph
and 
modify it to satisfy the assumptions made in Section~\ref{sec:nonzeroCirc}.
All these constructions are in log-space.

\subsection{Biconnected Graphs}
\label{sec:biconnected}
If a graph $G$ is disconnected then a perfect matching in $G$
can be constructed by taking a 
union of perfect matchings in its different connected components. 
As connected components of a graph can be found log-space \cite{Rei08}, 
we will always assume that the given graph is connected. 

Let $G$ be a connected graph. 
A vertex $a$ in $G$ is called an articulation point, if
its removal will make $G$ disconnected. 
A graph without any articulation point is called biconnected. 
Let $a$ be an articulation point in $G$ such that its deletion creates
connected components $G_1, G_2, \dots, G_m$.  
It is easy to see that for $G$ to have a perfect matching, 
exactly one of these components should have odd number of vertices, say $G_1$.
Then, in any perfect matching of $G$, the vertex $a$ will always be matched to a vertex 
in $G_1$. 
Thus, we can delete any edge connecting $a$ to other components,
and all the perfect matchings will still be preserved. 
It is easy to see that finding all the articulation points
and for each articulation point, performing the above mentioned reduction can be done in log-space,
via reachability queries \cite{Rei08,TW14}. 
Thus, we will always assume that the given graph is biconnected. 

\subsection{Matching Preserving Operation}

\begin{figure}
\centering
  \begin{picture}(0,0)%
  \includegraphics{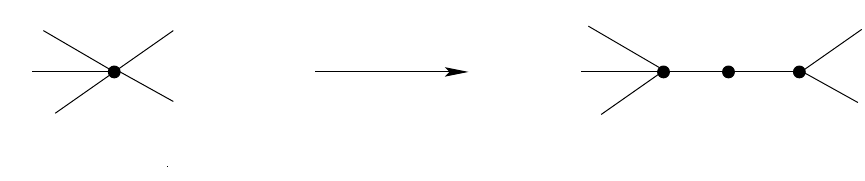}%
  \end{picture}%
  \input{vertexSplit.figtex}%

\caption{Vertex-Split: A vertex $v$ is split into three vertices $v,v',v''$, which are connected by a path. 
Some of the edges incident on $v$ are transferred to $v''$.  }
\label{fig:vertexSplit}
\end{figure}

{\bf Vertex-Split:} For a graph $G$, we define an operation called {\em vertex-split}, 
which {\em preserves matchings}, as follows:
Let $v$ be a vertex and let $X$ be the set of all the edges incident on $v$. 
Let $X_1 \sqcup X_2$ be an arbitrary partition of $X$. 
Create two new vertices $v'$ and $v''$ (see Figure~\ref{fig:vertexSplit}).
 Make the edges $(v,v')$ and $(v',v'')$. 
We call these two edges as {\em auxiliary edges}. 
For all the edges in $X_2$,
change their endpoint $v$ to $v''$.
We denote this operation by vertex-split$(v,X_1,X_2)$.

Let the modified graph be $G'$. 
One can go back to the graph $G$ by identifying vertices $v$, $v'$ and $v''$
and deleting auxiliary edges.
This operation is {\em matching preserving} in the following sense.

\begin{lemma}
There is a one-one correspondence between perfect matchings of $G$ and $G'$.
\label{lem:vertexSplit}
\end{lemma}
\begin{proof}
Consider a perfect matching $M$ in $G$, where $v$ is matched with a vertex in 
$X_1$. It is easy to see that the matching $M' := M \cup \{ (v', v'')\}$ is a perfect
matching in $G'$. 
The other case when $v$ is matched with a vertex in $X_2$ is similar.

Consider a perfect matching $M'$ in $G'$. 
Removing the auxiliary edge from $M'$ and identifying the vertices $v$, $v'$ and
$v''$ will give us a perfect matching in $G$.
\end{proof}

\subsection{Component Tree}
Wagner \cite{Wag37} and Asano \cite{Asa85} gave exact characterizations 
of $K_5$-free graphs and $K_{3,3}$-free graphs, respectively.
These characterizations
essentially mean that any graph in these two classes
 can be constructed by taking $3$-clique-sums
of graphs which are either planar or have size bounded by $8$.

\begin{theorem}{\cite{Asa85}}
Let $\mathcal{C}$ be the class of all planar graphs together with 
the $5$-vertex clique $K_5$. Then $\chev{\mathcal{C}}_2$ is the class
of $K_{3,3}$-free graphs. 
\end{theorem}

\begin{theorem}{\cite{Wag37, Khu88}}
Let $\mathcal{C}$ be the class of all planar graphs together with 
the four-rung M\"{o}bius ladder $V_8$ (Figure~\ref{fig:v8}). Then $\chev{\mathcal{C}}_3$ is the class
of $K_{5}$-free graphs. 
\end{theorem}
\begin{figure}
\centering
  \begin{picture}(0,0)%
  \includegraphics{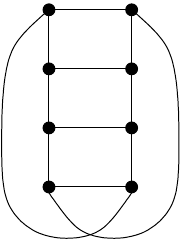}%
  \end{picture}%
  \input{v8.figtex}%

\caption{The four-rung M\"{o}bius ladder $V_8$.}
\label{fig:v8}
\end{figure}

As mentioned in Section~\ref{sec:biconnected}, we can assume that the given graph
is biconnected. 
It is known that for any given biconnected $K_{3,3}$ graph $G$,
its component tree can be constructed in log-space \cite[Lemma 3.8]{TW14}. 
The components here are all planar or $K_5$, which share separating pairs. 
Also, for any given biconnected $K_5$-free graph $G$, 
its component tree can be constructed in log-space \cite[Definition 5.2, Lemma 5.3]{STW14}. 
The components here are all planar or $V_8$.
They can share a separating pair or a separating triplet. 
The planar embedding of a planar component can be computed in log-space \cite{AM04,Rei08}. 

The component tree defined in \cite{TW14,STW14} slightly differs from 
our definition in Section~\ref{sec:componentTree}. 
They have an extra component for each separating set.
This component is connected to all the components which share this separating set. 
Moreover, whenever there is a real edge between two nodes of a separating set,
it is represented by a $3$-bond component (one real edge and two parallel virtual edges).
The $3$-bond component is also connected to the corresponding separating set node.
For our purposes, these two kinds of components are not needed.

For any given biconnected $K_{3,3}$-free graph or $K_5$-free graph $G$,
we start with the component trees which are constructed by \cite{TW14,STW14}.
We show how to modify the component tree, in log-space, 
to have 
the assumptions made in Section~\ref{sec:nonzeroCirc}.

Applying the clique-sum operations on the modified component tree will give us the actual modified graph $G'$.
We will argue that all these modifications in $G$
are just repeated application of the vertex-split operation (Lemma~\ref{lem:vertexSplit})
in $G$. Thus, these are matching preserving. As mentioned earlier,
from a perfect matching in $G'$, one can get a perfect matching in $G$
by just deleting the auxiliary vertices and edges created in the vertex-split operations.

We reiterate here that there may be some pairs/triplet in the graph $G$ (or $G'$), such 
that their removal will make the graph disconnected, but still the graph is not decomposed
with respect to them and they do not play any role in the component tree. 
Here, by separating pair/triplet we only mean those pairs/triplets which
are shared by different components of the component tree. 

\paragraph{(i) Removing ``$3$-bond" components:}
For all the $3$-bond components we do the following:
Remove the $3$-bond component. Let $\tau$ be the separating set and
$C_\tau$ be the corresponding node in the component tree,
where this $3$-bond component is attached (a $3$-bond component is always a leaf).
Take an arbitrary component attached to $C_\tau$.
This component will have a virtual clique for $\tau$. 
 Make an appropriate real edge parallel to the existing virtual edge,
in this virtual clique corresponding to $\tau$.
Note that if this component was planar, it will remain so.
Moreover, it is easy to adjust the planar embedding. 
Clearly, this operation can be done in log-space.
This does not change the actual graph $G$ in any way.

\paragraph{(ii) Any separating set is shared by at most two components:}
Let $\tau$ be a separating set shared by $m$ components $G_1, G_2, \dots, G_m $.
Let the cardinality of $\tau$ is $t$ (t can be 2 or 3). 
Let us define a gadget $M$ as follows: 
it has three sets of nodes $\{a_i \mid 1 \leq i \leq t\}$, 
$\{b_i \mid 1 \leq i \leq t\}$, $\{c_i \mid 1 \leq i \leq t\}$. 
For each $i$, connect $a_i$ with $b_i$ by a length-$2$ path and also 
connect $a_i$ with $c_i$ by a length-$2$ path. 
Make $3$ virtual cliques each of size $t$,
one each for nodes $\{a_i\}_i$, $\{b_i\}_i$ and $\{c_i\}_i$.
Thus, three components can be attached with $M$.

Now, we construct a binary tree $T$ which has exactly $m-1$ leaves. 
Replace leaves of $T$ with components $G_2, G_3, \dots, G_m$.
Replace all other nodes of $T$ with copies of the gadget $M$.
Further, make an edge between component $G_1$ and the root of $T$ 
(see Figure~\ref{fig:binaryTree}).
Any node of type $M$, in this binary tree, shares its separating set $\{a_i\}_i$
with its parent node, shares its separating set $\{b_i\}_i$
with its left child node and shares its separating set $\{c_i\}_i$
with its right child node. 
The components $G_2,G_3, \dots, G_m$ share their copy of $\tau$ with
their respective parent nodes in the tree $T$.
The component $G_1$ shares its copy of $\tau$ with the root node of $T$.

\begin{figure}
\centering
  \begin{picture}(0,0)%
  \includegraphics{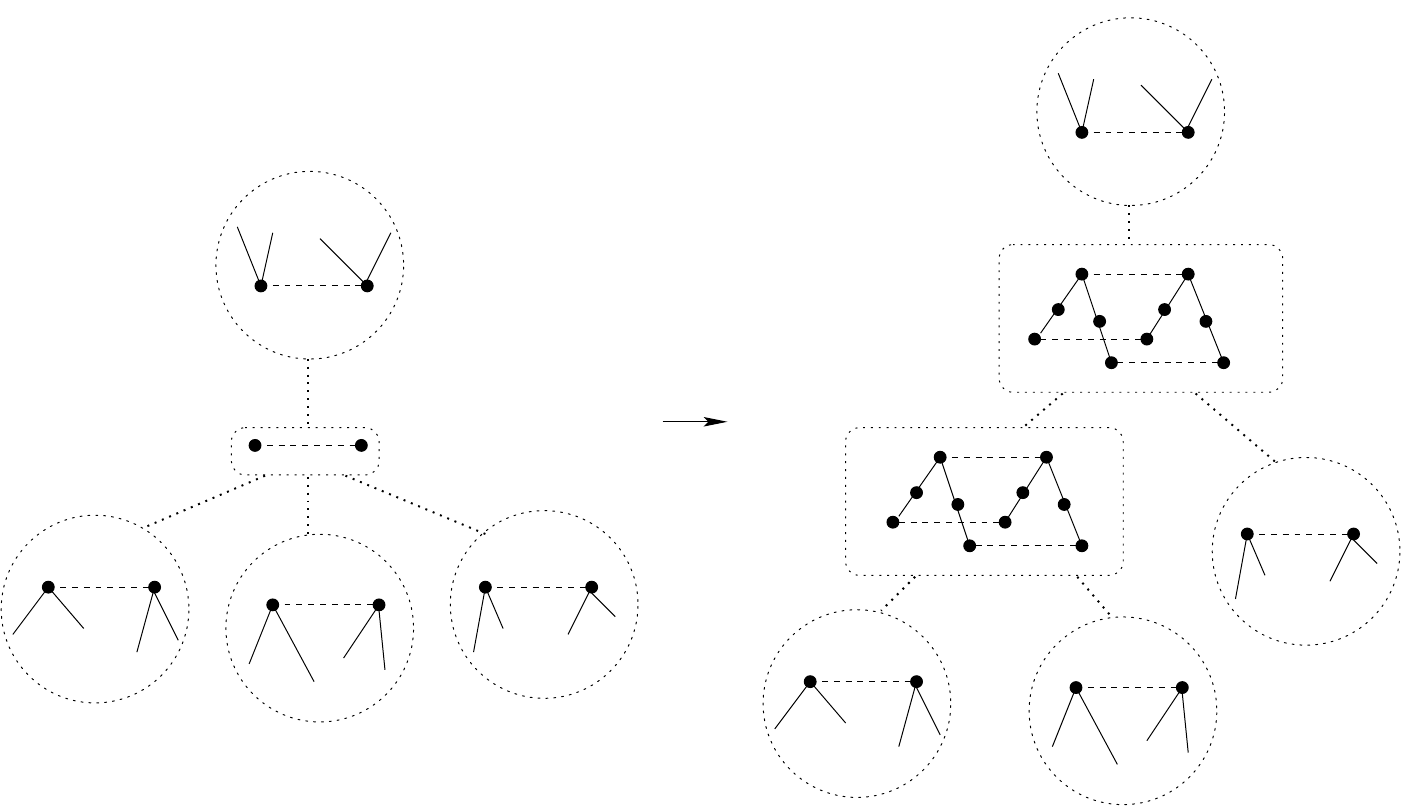}%
  \end{picture}%
  \input{binaryTree.figtex}%

\caption{(a) A separating pair $\chev{u,v}$ is shared by four components $G_1, G_2, G_3, G_4$.
(b) Copies of $\chev{u,v}$ connected by length-$2$ paths, to form a binary tree. 
Different copies are shared by different components. 
  }
\label{fig:binaryTree}
\end{figure}

Doing this procedure for every separating set will ensure that every separating
set is shared between at most two components.
Moreover, now there is no extra component for the separating set, and
the components which share a separating set are joined directly by an edge.
A binary tree with $m-1$ leaves can be easily constructed in log-space
(Take nodes $\{x_1,x_2, \dots, x_{2m-3}\}$, $x_i$ has children $x_{2i}$ and $x_{2i+1}$).
All the other operations here are local like deleting and creating edges and 
changing vertex labels. Thus it can be done log-space.

Now, we want to argue that this operation is matching preserving for 
the actual graph $G$. 
Let us view this operation as a repeated application of the following operation:
Partition the the set of components $\{G_2, G_3, \dots G_m\}$ in two parts,
say $G'_1$ and $G''_1$. Now, take a copy of the gadget $M$ and connect it to
all three components $G_1$, $G'_1$ and $G''_1$. 
$M$ shares its separating sets $\{a_i\}_i$, $\{b_i\}_i$ and $\{c_i\}_i$
with $G_1$, $G'_1$ and $G''_1$ respectively. 
In the actual graph $G$, this operation separates the edges incident on a vertex in $\tau$
into three parts: edges from $G_1$, $G'_1$ and $G''_1$ respectively.
These three sets of edges are now incident on three different copies of the vertex.
Moreover two of the copies are connected to the first copy via a length-$2$ path.
Hence, it is easy to see this as applying vertex-split (Lemma~\ref{lem:vertexSplit}) operation twice. 
Now, we recursively do the same operation after partitioning the
set of components $G'_1$ and $G''_1$ further. 
Thus, the whole operation can be seen as a vertex-split operation applied many times
in the actual graph $G$. 

Instead of a binary tree we could have also taken 
a tree with one root and $m-1$ leaves. This operation would also be
matching preserving but the component size will depend on $m$. 
On the other hand, in our construction the new components created have
size at most $15$ (number of real edges is bounded by 12). 
Thus, the graph $G'$ remains in class $\Gt$.

\paragraph{(iii) Any vertex is a part of at most one separating set:}
Let $a$ be vertex in a component $C$, where it is a part of
separating sets $\tau_1, \tau_2, \dots, \tau_m$.
We apply the vertex-split operation (Lemma~\ref{lem:vertexSplit})
 on $a$, $m$ times, to split $a$ into a star.
Formally, create a set of $m$ new nodes $a_1, a_2, \dots, a_m$. 
Connect each $a_i$ with $a$ by a path of length $2$. 
For each $i$, replace $a$ with $a_i$ in the separating set $\tau_i$.
Let the updated separating set be $\tau'_i$.
The edge in the component tree which corresponds to $\tau_i$,
should now correspond to $\tau'_i$.
Any real edge in the component $C$ which is incident on $a$, 
remains that way (see Figure~\ref{fig:star}).
Clearly, doing this for every vertex in all the components will ensure
that every vertex is a part of at most one separating set. 

It is easy to see that a planar component will remain planar after this operation. 
The modification of the planar embedding and other changes here 
are local and can be done in log-space.

\begin{figure}
\centering
  \begin{picture}(0,0)%
  \includegraphics{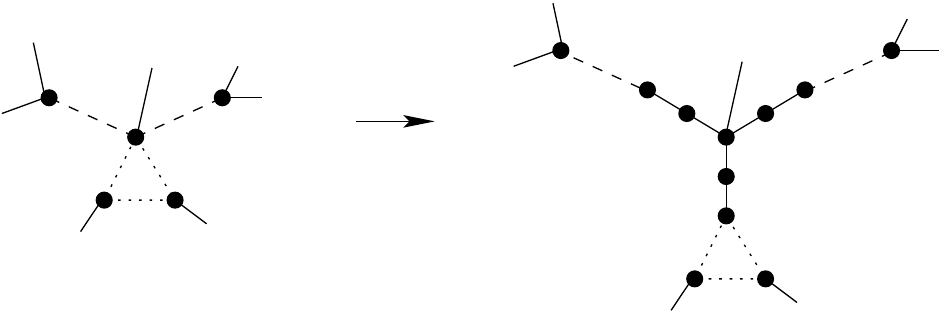}%
  \end{picture}%
  \input{star.figtex}%

\caption{(a) Vertex $a$ is a part of two separating pairs $\chev{a,d}$ and $\chev{a,e}$
and a separating triplet $\chev{a,b,c}$.
 (b) Vertex-Split is applied on vertex $a$, $3$ times, to split it into a star. 
The new separating sets are $\chev{a_1,b,c}$, $\chev{a_2,d}$ and $\chev{a_3,e}$.
}
\label{fig:star}
\end{figure}

Now, we want to argue that this operation is matching preserving. 
Let us see how does this operation modifies the actual graph $G$.
Let $C_i$ be the component which shares $\tau_i$ with $C$. 
Removal of $\tau_i$ would split the graph $G$ into two components,
say $G'_i$ and $G''_i$, where $G'_i$ is the one containing $C$.
The above operation means that any edge in $G''_i$ which was 
incident on $a$, is now incident on $a_i$ instead of $a$.
As each $a_i$ is connected to $a$ by a length-$2$ path, 
this operation can be seen as a repeated application of the
vertex-split operation (Lemma~\ref{lem:vertexSplit}). 
Thus, this operation is matching preserving. 

{\em Increase in the size of non-planar components:}
After this operation the size of each component will grow. 
Let us find out the new bound on the size of constant-sized graphs.
For a $K_{3,3}$-free graph, all non-planar components are of type $K_5$.
Moreover, they are only involved in a $2$-clique-sum. 
Hence, it can have at most ${5 \choose 2} =10$ separating pairs. 
In this case, each vertex is a part of four separating pairs. 
Thus, each vertex will be split into a $4$-star, creating $8$ new vertices
and $8$ new edges. Totally, there will be $45$ vertices and 
$40$ real edges. Additionally, there can be some already existing real edges,
at most $10$. Thus, the total number of edges is bounded by $50$.

For a $K_5$-free graph, all non-planar components are of type $V_8$.
Moreover, they do not have a $3$-clique, thus, can only be be involved in 
a $2$-clique-sum. 
In worst case, it has $12$ separating pairs. 
Each vertex is a part of $3$ separating pairs. 
Each vertex will be split into a $3$-star, creating $6$ new vertices
and $6$ new edges. Totally, there will be $56$ vertices and $48$ edges. 
Thus, together with already existing real edges, total number of real 
edges is bounded by $60$.

\paragraph{(iv) A separating triplet in a planar component already forms a face:}
If a separating triplet does not form a face in a planar component. Then
the two parts of the graph, one inside the triplet and the other outside,
can be considered different components sharing this triplet. 
In fact, the construction in \cite{STW14} already does this. 
When they decompose a graph with respect to a triplet, 
the different components one gets by deleting this triplet are all considered
different components in the component tree.

\section{Discussion}
One of the open problems is to construct an isolating weight assignment
for a more general class of graphs, in particular, for all bipartite graphs. 
Note that {\em nonzero circulation} for every cycle is sufficient but not necessary for constructing
an isolating weight assignment. 
Although existence of an isolating weight assignment can be shown by randomized arguments,
no such arguments exist for showing the existence of a nonzero circulation weight assignment. 
It needs to be investigated whether it is possible to achieve a nonzero circulation for every cycle 
(with polynomially bounded weights)
in a complete bipartite graph?
Log-space construction of such a weight assignment would imply 
that Bipartite Perfect Matching is in $\NC$ and answer the $\NL$=$\UL$? question.

Till now, isolation of a perfect matching is known only for those graphs
for which counting the number of perfect matchings is easy. 
On the other hand, $O(\log n)$-genus bipartite graphs 
and general planar graphs are two classes of graphs
for which counting is easy, but 
construction of an isolating weight assignment
 is not known. 
It is surprising, as counting seems to be a much harder problem than isolation.

\section{Acknowledgements}
RG thanks TCS PhD research fellowship for support.

\bibliographystyle{alpha}
\bibliography{k33free}

\begin{thebibliography}{DKTV12}

\bibitem[AHT07]{AHT07}
Manindra Agrawal, Thanh~Minh Hoang, and Thomas Thierauf.
\newblock The polynomially bounded perfect matching problem is in {NC$^2$}.
\newblock In Wolfgang Thomas and Pascal Weil, editors, {\em STACS 2007}, volume
  4393 of {\em Lecture Notes in Computer Science}, pages 489--499. Springer
  Berlin Heidelberg, 2007.

\bibitem[AM04]{AM04}
Eric Allender and Meena Mahajan.
\newblock The complexity of planarity testing.
\newblock {\em Information and Computation}, 189(1):117 -- 134, 2004.

\bibitem[ARZ99]{ARZ99}
Eric Allender, Klaus Reinhardt, and Shiyu Zhou.
\newblock Isolation, matching, and counting uniform and nonuniform upper
  bounds.
\newblock {\em J. Comput. Syst. Sci.}, 59(2):164--181, 1999.

\bibitem[Asa85]{Asa85}
Takao Asano.
\newblock An approach to the subgraph homeomorphism problem.
\newblock {\em Theoretical Computer Science}, 38(0):249 -- 267, 1985.

\bibitem[BTV09]{BTV09}
Chris Bourke, Raghunath Tewari, and N.~V. Vinodchandran.
\newblock Directed planar reachability is in unambiguous log-space.
\newblock {\em ACM Trans. Comput. Theory}, 1(1):4:1--4:17, February 2009.

\bibitem[DKR10]{DKR10}
Samir Datta, Raghav Kulkarni, and Sambuddha Roy.
\newblock Deterministically isolating a perfect matching in bipartite planar
  graphs.
\newblock {\em Theory of Computing Systems}, 47:737--757, 2010.

\bibitem[DKTV12]{DKTV12}
Samir Datta, Raghav Kulkarni, Raghunath Tewari, and N.~V. Vinodchandran.
\newblock Space complexity of perfect matching in bounded genus bipartite
  graphs.
\newblock {\em J. Comput. Syst. Sci.}, 78(3):765--779, 2012.

\bibitem[Edm65]{Edm65}
Jack Edmonds.
\newblock Path, trees, and flowers.
\newblock {\em Canadian J. Math.}, 17:449–467, 1965.

\bibitem[GK87]{GK87}
Dima Grigoriev and Marek Karpinski.
\newblock The matching problem for bipartite graphs with polynomially bounded
  permanents is in {NC} (extended abstract).
\newblock In {\em 28th Annual Symposium on Foundations of Computer Science, Los
  Angeles, California, USA, 27-29 October 1987}, pages 166--172, 1987.

\bibitem[GL99]{GL99}
Anna Galluccio and Martin Loebl.
\newblock On the theory of {Pfaffian} orientations. {I}. perfect matchings and
  permanents.
\newblock {\em Electr. J. Comb.}, 6, 1999.

\bibitem[Hoa10]{Hoa10}
Thanh~Minh Hoang.
\newblock On the matching problem for special graph classes.
\newblock In {\em IEEE Conference on Computational Complexity}, pages 139--150.
  IEEE Computer Society, 2010.

\bibitem[HT73]{HT73}
John~E. Hopcroft and Robert~Endre Tarjan.
\newblock Dividing a graph into triconnected components.
\newblock {\em SIAM J. Comput.}, 2(3):135--158, 1973.

\bibitem[Jor69]{Jor69}
Camille Jordan.
\newblock Sur les assemblages de lignes.
\newblock {\em Journal für die reine und angewandte Mathematik}, 70:185--190,
  1869.

\bibitem[Khu88]{Khu88}
S.~Khuller.
\newblock {\em Parallel Algorithms for {$K_5$}-minor Free Graphs}.
\newblock Cornell University, Department of Computer Science, 1988.

\bibitem[KMV08]{KMV08}
Raghav Kulkarni, Meena Mahajan, and Kasturi~R. Varadarajan.
\newblock Some perfect matchings and perfect half-integral matchings in {NC}.
\newblock {\em Chicago Journal of Theoretical Computer Science}, 2008(4),
  September 2008.

\bibitem[Kor09]{Kor09}
Arpita Korwar.
\newblock Matching in planar graphs.
\newblock Master's thesis, Indian Institute of Technology Kanpur, 2009.

\bibitem[KUW86]{KUW86}
Richard~M. Karp, Eli Upfal, and Avi Wigderson.
\newblock Constructing a perfect matching is in random {NC}.
\newblock {\em Combinatorica}, 6(1):35--48, 1986.

\bibitem[Lin92]{Lin92}
Steven Lindell.
\newblock A logspace algorithm for tree canonization (extended abstract).
\newblock In {\em Proceedings of the Twenty-fourth Annual ACM Symposium on
  Theory of Computing}, STOC '92, pages 400--404, New York, NY, USA, 1992. ACM.

\bibitem[LMR07]{LMR07}
Nutan Limaye, Meena Mahajan, and B.V.Raghavendra Rao.
\newblock Arithmetizing classes around {NC$_1$} and l.
\newblock In {\em STACS 2007}, volume 4393 of {\em Lecture Notes in Computer
  Science}, pages 477--488. Springer Berlin Heidelberg, 2007.

\bibitem[Lov79]{Lov79}
L{\'a}szl{\'o} Lov{\'a}sz.
\newblock On determinants, matchings, and random algorithms.
\newblock In {\em FCT}, pages 565--574, 1979.

\bibitem[MV80]{MV80}
Silvio Micali and Vijay~V. Vazirani.
\newblock An {$O(\sqrt{V}E)$} algorithm for finding maximum matching in general
  graphs.
\newblock In {\em Proceedings of the 21st Annual Symposium on Foundations of
  Computer Science}, SFCS '80, pages 17--27, Washington, DC, USA, 1980. IEEE
  Computer Society.

\bibitem[MVV87]{MVV87}
Ketan Mulmuley, Umesh~V. Vazirani, and Vijay~V. Vazirani.
\newblock Matching is as easy as matrix inversion.
\newblock {\em Combinatorica}, 7:105--113, 1987.

\bibitem[NTS95]{NT95}
Noam Nisan and Amnon Ta-Shma.
\newblock Symmetric logspace is closed under complement.
\newblock In Frank~Thomson Leighton and Allan Borodin, editors, {\em STOC},
  pages 140--146. ACM, 1995.

\bibitem[RA00]{RA00}
Klaus Reinhardt and Eric Allender.
\newblock Making nondeterminism unambiguous.
\newblock {\em {SIAM} J. Comput.}, 29(4):1118--1131, 2000.

\bibitem[Rei08]{Rei08}
Omer Reingold.
\newblock Undirected connectivity in log-space.
\newblock {\em J.\ ACM}, 55:17:1--17:24, September 2008.

\bibitem[STW14]{STW14}
Simon Straub, Thomas Thierauf, and Fabian Wagner.
\newblock Counting the number of perfect matchings in k5-free graphs.
\newblock In {\em {IEEE} 29th Conference on Computational Complexity, {CCC}
  2014, Vancouver, BC, Canada, June 11-13, 2014}, pages 66--77, 2014.

\bibitem[TV12]{TV12}
Raghunath Tewari and N.~V. Vinodchandran.
\newblock Green's theorem and isolation in planar graphs.
\newblock {\em Inf. Comput.}, 215:1--7, 2012.

\bibitem[TW14]{TW14}
Thomas Thierauf and Fabian Wagner.
\newblock Reachability in {$K_{3,3}$}-free and {$K_5$}-free graphs is in
  unambiguous logspace.
\newblock {\em Chicago J. Theor. Comput. Sci.}, 2014, 2014.

\bibitem[Vaz89]{Vaz89}
Vijay~V. Vazirani.
\newblock {NC} algorithms for computing the number of perfect matchings in
  {${K}_{3,3}$}-free graphs and related problems.
\newblock {\em Information and Computing}, 80(2):152--164, 1989.

\bibitem[Wag37]{Wag37}
Klaus Wagner.
\newblock {\"{U}ber eine Eigenschaft der ebenen Komplexe}.
\newblock {\em Math. Ann.}, 114, 1937.

\end{thebibliography}

\appendix
\section{Skipped proofs}
\label{sec:appendix}

Here we prove the lemmas whose proofs were skipped in the main part of the paper.
\begin{lem:circulationFaces}
In a planar graph with a given planar embedding, circulation of a cycle in the clockwise orientation 
is the sum of circulations of the faces inside it.
\end{lem:circulationFaces}
\begin{proof}
We give the proof using mathematical induction on the number of faces inside the cycle.

Consider a planar graph $G=(V,E)$.
For any cycle $C$, its circulation is denoted by $w(C)$.

\textbf{Base case:}
The base case is a cycle containing only one face inside it. 
By definition of the circulation of a face, 
for a clockwise-oriented cycle, its circulation equals the circulation of the face. 

\textbf{Induction hypothesis:} The circulation of a cycle having $k$ faces is the sum of circulations of the faces inside it.

\textbf{Induction step:}
Consider a clockwise-oriented cycle $C$ having $k$~faces, $\allowbreak f_1, f_2, \allowbreak\dots, f_k$, inside it.
Now consider a cycle $C'$ having the same orientation as $C$ and with all but one face of $C$ inside it. Without loss of generality, let this face be $f_k$.

We use the notation $E_{ij}$ to show the set of edges shared between faces $f_i$ and $f_j$, taken in a clockwise direction around $f_i$.

Denote by $S_k$ the set of clockwise edges (w.r.t $f_k$) shared between $f_k$ and other faces inside $C$, that is, $S_k = \cup_{i=0}^{k-1}E_{ki}$.
Let $S_{-k}$ denote the same set of edges taken in the opposite direction.

Also, we use the notation $E(C)$ to denote the set of edges taken by a cycle $C$.
Similarly, $E_k$ denotes the set of edges around a face $f_k$, taken in the clockwise direction.
Similar to $S_{-k}$, we can define $E_{-k}$.

We can see that $E(C)\setminus E(C') = E_k\setminus S_k$, and $E(C')\setminus E(C) = S_{-k}$.
\begin{align*}
w(C) &= w(C') + w(E(C)\setminus E(C')) -  w(E(C')\setminus E(C))\\
&= w(C') + w(E_k\setminus S_k) - w(S_{-k})\\
&= w(C') + w(E_k) - w(S_k) - w(S_{-k})\\
&= w(C') + w(E_k) - w(S_k) + w(S_{k}) &&\text{($w$ is skew-symmetric)}\\
&= \sum_{i=1}^{k-1}w(f_i) + w(E_k) &&\text{(Induction hypothesis)}\\
&= \sum_{i=1}^{k-1}w(f_i) + w(f_k) &&\text{(Lemma \ref{lem:faceToEdge})}
\end{align*}

Thus, the circulation of $C$ is the sum of circulations of the faces contained in it.
\end{proof}

\begin{lem:faceToEdge}
Let $G(V,E)$ be a planar graph with $F$ being its set of inner faces in some planar embedding.
For any given function on the inner faces $w' : F \to \Z$, a skew symmetric 
weight function $w \colon \vec{E} \to \Z$ can be constructed in log-space such that every face $f \in F$
has circulation $w'(f)$.
\end{lem:faceToEdge}
\begin{proof}
The construction in \cite{Kor09} gives $+1$ circulation to every face of the graph and 
is in $\NC$. We modify it to assign arbitrary circulations to the faces and argue that it
works in log-space.

 Let $G^*$ be the dual graph of $G$ and $T^*$ be a spanning tree of $G^*$.
The dual graph can be easily constructed in log-space from the planar embedding. 
See \cite{NT95,Rei08}  for log-space construction of a spanning tree.
Make the tree $T^*$ rooted at the outer face of $G$.
All the edges in $E \setminus E(T^*)$ will get weight $0$.
For any node $f$ in $G^*$ (a face in $G$),
let $T_f^*$ denote the subtree of $T^*$ rooted at $f$.
Let $w'(T_f^*)$ denote the total sum of the weights in the tree, 
i.e.\ $w'(T_f^*) = \sum_{f_1 \in T_f^*} w'(f_1)$. 
This function can be computed for every node in the tree $T^*$,
 by the standard log-space
tree traversal.
For any inner face $f$, 
let $e_f$ be the edge connecting $f$ to its parent in the dual tree $T^*$. 
We assign the edge $e_f$, weight $w'(T_f^*)$ in clockwise direction (w.r.t.\ face $f$).

We claim that under this weight assignment, circulation of any inner face $f$ 
is $w'(f)$.
To see this, let us say $f_1, f_2, \dots, f_k$ are the children of $f$ in the dual tree $T^*$.
These nodes are connected with $f$ using edges $e_{f_1},e_{f_2}, \dots, e_{f_k}$ respectively. 
Now, consider the weights of these edges in the clockwise direction w.r.t.\ face $f$.
For any $1 \leq i \leq k$, weight of $e_{f_i}$ is $- w'(T_{f_i}^*)$ and
weight of $e_f$ is $w'(T_{f}^*)$.
Clearly, sum of all these weights is $w'(f)$.
\end{proof}

\end{document}